\documentclass[10pt,conf]{IEEEtran}

\usepackage{graphicx}
\usepackage{amsmath}
\usepackage[ruled,vlined,linesnumbered]{algorithm2e}
\SetKwProg{Fn}{Function}{}{}

\usepackage{subfigure}
\usepackage{mathrsfs}
\usepackage{amsthm}
\usepackage{amssymb}
\usepackage{color}
\usepackage{hyperref}
\usepackage[columnwise,displaymath,mathlines]{lineno}
\usepackage{multirow}
\usepackage{booktabs}

\usepackage{xcolor}
\usepackage{colortbl}

\newtheorem{theorem}{Theorem}
\newtheorem{definition}{Definition}
\newtheorem{lemma}{Lemma}

\begin{document}
%
\title{IMF: Iterative Max-Flow for Node Localizability Detection in Barycentric Linear Localization}
%
%
%
%

\author{Haodi~Ping,  Yongcai~Wang, 
        and~Deying~Li
\IEEEcompsocitemizethanks{\IEEEcompsocthanksitem The authors are with School of Information, Renmin University of China, Beijing, P.R.China, 100872.\protect\\}
\thanks{E-mail: \{haodi.ping, ycw, deyingli\}@ruc.edu.cn}
}

\IEEEtitleabstractindextext{%
\begin{abstract}

Determining whether nodes can be uniquely localized, called localizability detection, is a concomitant problem of network localization. Localizability under traditional Non-Linear Localization (NLL) schema has been well explored, whereas localizability under the emerging Barycentric coordinate-based Linear Localization (BLL) schema has not been well touched. In this paper, we investigate the deficiency of existing localizability theories and algorithms in BLL, and then propose a necessary condition and a sufficient condition for BLL node localizability. Based on these two conditions, an efficient iterative maximum flow (IMF) algorithm is designed to identify BLL localizable nodes. Finally, our algorithms are validated by both theoretical analysis and experimental evaluations. 

\end{abstract}

}

\maketitle

\IEEEdisplaynontitleabstractindextext

%
\IEEEpeerreviewmaketitle

\section{Introduction}\label{sect:intro}

With the rapid development of communication technologies such as fifth-generation (5G) networks, and the growing popularity of smart objects such as sensors, robots, unmanned aerial vehicles (UAVs), there is potential for networking heterogeneous devices to provide services. In these applications, the geographical locations of the objects are one of the most fundamental knowledge \cite{Guo2020,Nguyen2020,DEMIANE2020107378,bartoletti20185g}. Thus \emph{Network Localization (NL)}, which calculates the geographical node locations using the inter-node measurements and anchor information, has attracted tremendous both academic interest \cite{guo2020simultaneous,yu2019novel,zhao2019convolutional} and commercial concerns\cite{fira,tsingoal,infsoft}. 

Existing NL algorithms can be broadly divided into two categories. Traditionally, NL is modeled as a non-linear least-square optimization problem, whose objective is to minimize the sum square errors between the estimated distances (calculated by node locations) and the measured distances. We call these algorithms \emph{Non-Linear Localization} (NLL) algorithms. NLL can be solved by multi-dimensional scaling \cite{shang2004improved},  graph optimization \cite{PingHGO,bhat2020optimization}, and component stitching \cite{Ping2020,Sun2018}, etc. However, NLL algorithms need to be well initialized, otherwise, it is easy to fall into a local minimum. NLL is also not efficient and is difficult to be distributed for it is essentially a quadratic equation.  To overcome the defects of NLL algorithms, a more recent Barycentric coordinate-based Linear Localization (BLL) framework has emerged. In BLL, each node represents its location as a linear combination of locations of neighbors through the barycentric coordinates. Then the network localization problem reduces to a linear model, which is computationally efficient, easy to be distributed, and is guaranteed convergence under broad initialization \cite{Khan2015,linear_survey,DILOC,ECHO,Tecchio2019NDimensional}. 

Along with the evolution of NL algorithms, there is a concomitant problem called \emph{localizability detection}. \textbf{Network localizability} characterizes whether an entire network is uniquely localizable given the distance and anchor constraints while \textbf{node localizability} answers whether a specific node is localizable. Prior to calculating node locations, localizability is of great significance to the feasibility of NL algorithms, which helps NL algorithms avoid wrong location answers\cite{xia2016localizability}. The localizability property also provides guidelines for network topology control, node deployment optimization, and event detection. 

The network localizability identification problem was usually researched through graph rigidity theories\cite{eren2004rigidity,hendrickson1992conditions,jackson2005connected}. Based on the rigidity theory, the necessary and sufficient condition for network localizability of NLL methods has been proposed and a polynomial algorithm for localizability testing has been designed\cite{jackson2005connected}. However, considering the general requirement of deployment cost and limited perception radius for energy efficiency, practical networks are usually sparsely deployed. For sparse networks, the network localizability can hardly be satisfied, because it is inevitable that some sparsely connected nodes are unlocalizable. 
The previous study has shown that real networks tend to be not entirely localizable, but a certain portion of the nodes can still be uniquely localized\cite{goldenberg2005network}. Thus, it is more meaningful to find out how many nodes can be uniquely localized and which are them. This problem is called \emph{node localizability detection}. 
Herein, we should be clear about the concept of \emph{network localizability,} and \emph{node localizability}. Detecting localizable nodes deserve great research attention in both NLL and BLL.   

The research on node localizability includes the node localizability conditions and the node localizability detection algorithms. In NLL, Yang et. al.\cite{yang2011understanding} proposed an RR3P condition, which is so far the best sufficient condition for node localizability. But up to now, there is no known necessary and sufficient condition for node localizability. For the node localizability detection algorithms, a node localizability detection algorithm was designed according to the RR3P condition. Moreover, several incremental detection algorithms have been proposed according to the rigidity theory\cite{eren2004rigidity,Yang2009,wu2017triangle}. In BLL, Diao et.al. proposed the sufficient and necessary condition for network localizability\cite{ECHO}. But the node localizability condition and node localizability detection algorithms for BLL are still open.

Non-awareness of node localizability is a critical problem in BLL, because each node's location is iteratively updated based on neighbors' locations in the BLL routine. If non-localizable nodes participate in the location propagation process, their wrong locations will impact the locations of the localizable nodes. Fig.~\ref{fig:show_LABLL_2d} shows an example of locating 12 nodes including 2 non-localizable nodes using ECHO\cite{ECHO}, a representative BLL localization algorithm. Let $\mathbf{p}_i$ denote the ground truth location of node $v_i$, and $\mathbf{\hat{p}}_i$ denote the estimated location using ECHO. Fig.~\ref{fig:NL_ECHO} shows the results when all the 12 nodes participate in BLL, where the asterisk represents $\mathbf{\hat{p}}_i$ and the circle represents $\mathbf{p}_i$. How the localization errors ($||\mathbf{\hat{p}}_i-\mathbf{p}_i||_2$) of the nodes vary with the iteration times is shown in Fig.~\ref{fig:err_echo_2d}. The results show that even after $3\times 10^5$ iterations, no node can converge to the correct location if localizability is not considered. Therefore, detecting localizable nodes and excluding non-localizable nodes are critically important for BLL. 

\begin{figure}[ht!]
	\centering
	\subfigure[ECHO in 2D.]{
		\begin{minipage}[t]{0.46\linewidth}
			\centering
			\includegraphics[width=.99\textwidth]{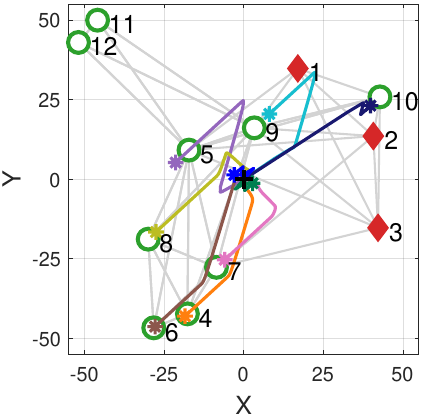}
			\label{fig:NL_ECHO}
		\end{minipage}
	}
\centering
	\subfigure[Error evolution]{
		\begin{minipage}[t]{0.46\linewidth}
			\centering
			\includegraphics[width=0.99\textwidth]{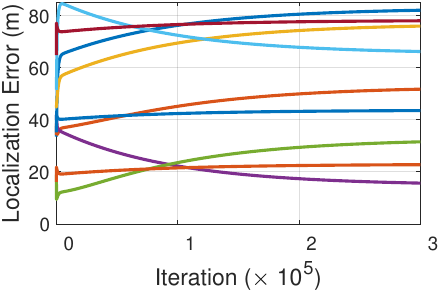}
			\label{fig:err_echo_2d}
		\end{minipage}
	}
	\caption{The convergence traces, where `$+$' represents the initialization of $\mathbf{\hat p}_i$, `$*$' represents the converged $\mathbf{\hat p}_i$, and `o' represents the ground truth location $\mathbf{p}_i$. (a) ECHO converges to wrong locations for all nodes. (b)  The evolution of localization error ($||\mathbf{\hat{p}}_i-\mathbf{p}_i||_2$) w.r.t. iteration rounds.}
	\label{fig:show_LABLL_2d}
\end{figure}

In this paper, we concentrate on localizability detection in BLL and its applications. The main contributions are as follows.

\begin{itemize}
	\item Firstly, a necessary condition and a sufficient condition are designed for BLL node localizability. 
	\item Then so far the first network localizability test algorithm for BLL is designed, whose key idea is to transform the detection of disjoint paths into calculating the max-flow. Moreover, an Iterative Max-Flow (IMF) method is designed to detect the BLL localizable nodes. 
	\item The time and space complexity of IMF are studied and the applicability of the IMF algorithms is verified in various aspects including its extension to NLL, extension to higher dimensions. 
\end{itemize}

The remainder of this paper is as follows. In Section~\ref{sect:preliminary}, an overview of the evolution of NL algorithms and the localizability theories is presented. In Section~\ref{sect:conditions}, a necessary condition and a sufficient condition for node localizability are proposed. In Section~\ref{sect:algorithms}, testing algorithms for network localizability and an iterative maximum flow (IMF) algorithm for identifying localizable nodes are proposed. Analysis and possible extensions of the proposed algorithms are given in Section~\ref{sect:analysis}. Algorithms are evaluated in Section~\ref{sect:evaluation}. Finally, this paper is concluded with discussion of future work in Section~\ref{sect:conclusion}.

\section{Preliminary and Key Related Work} \label{sect:preliminary}
The problem formulation, review of network localizability conditions in NLL and BLL, review of node localizability conditions in NLL and BLL, existing localizability identification algorithms, and an overview of our framework are introduced in this section.

\subsection{Notations and Problem Formulation}

For a network $\mathcal{G=\{V,E, \mathbf{d}\}}$, let $\mathcal{V}$ = $\mathcal{A} \cup \mathcal F$ denote the node set, which is composed of $m$ \emph{anchor nodes} (i.e., $\mathcal{A}=\{v_{m+1},\cdots,v_{m+n}\}$) and $n$ \emph{free nodes} (i.e., $\mathcal{F}=\{v_1,\cdots,v_m\}$). Locations of anchors are known, denoted by $\mathbf{P}_{\mathcal{A}}=\{\mathbf{p}_1,\cdots,\mathbf{p}_m\}$ while locations of free nodes are unknown, denoted by  $\mathbf{P}_{\mathcal{F}}=\{\mathbf{p}_{m+1},\cdots,\mathbf{p}_{m+n}\}$. $\mathbf d$ represents the distance matrix among nodes. For a node $v_i$, it can obtain the distance $d_{ij}$ to another node $v_j$ if $||\mathbf{p}_i - \mathbf{p}_i||_2\leq R$, where $R$ is the perception radius. $(i,j)\in\mathcal{E}$ if  $d_{ij}$ can be obtained. The neighbors of $v_i$ is represented as $\mathcal{N}_i$, $v_j\in \mathcal{N}_i$ if $(i,j)\in\mathcal{E}$. The problem is to detect which free nodes can be uniquely localized using the anchor and distance constraints when BLL methods are adopted. 

\subsection{BLL Methods, Localizability Conditions and Algorithms}



Range-based NL algorithms fall into two categories, i.e., NLL and BLL. In NLL, there are a wide range of algorithms to solve $\mathbf{P}_{\mathcal{F}}$ from the model of minimizing the sum square errors\cite{ji2013beyond,shang2004improved,so2007theory,kummerle2011g,Sun2018,Wang2018Formation,PingHGO,Ping2020}. The main defect of NLL is that it does not guarantee global optimality since the optimization model is quadratic. For a detailed introduction to NLL, please refer to reference \cite{Ping2020}. In this paper, we focus on BLL algorithms. 

\subsubsection{The Process of BLL}

In BLL, the basic idea is to represent a node's location as a linear combination of the locations of its neighbors using inter-node distances. Then, the network localization problem becomes a linear model. The detailed procedure for converting the NL problem to linear form is as follows.

In $\Re^2$, for a node $v_i$ and neighbors $v_j$, $v_k$, $v_l$ $\in\mathcal{N}_i$, the linear combination is: 

\begin{equation}
	\begin{aligned}
		&{\mathbf{p}}_i = a_{ij}{\mathbf{p}}_j + a_{ik}{\mathbf{p}}_k + a_{il}{\mathbf{p}}_l, 
	\end{aligned}
	\label{equ:barycentric_representation}
\end{equation}
where $\{a_{ij},a_{ik},a_{il}\}$ is called the \emph{barycentric coordinate} of $v_i$ introduced by $M\ddot{o}bius$\cite{Mobius1827}. Denote $\{v_j,v_k,v_l\}$ as $\Theta$. The barycentric coordinate is calculated as:
\begin{equation}
	\left\{
	\begin{aligned} 
		a_{ij} & = D(\Theta; v_i, v_k, v_l)/D(\Theta,\Theta),\\
		a_{ik} & = D(\Theta; v_j, v_i, v_l)/D(\Theta,\Theta), \\
		a_{il} & = D(\Theta; v_j, v_k, v_i)/D(\Theta,\Theta).
	\end{aligned}\right.
	\label{equ:barycentric_representation_bicm}
\end{equation}
$D(\cdot)$ represents the \emph{Cayley-Menger bideterminant}, whose input is two node sets with equal cardinality and the output is a signed scalar\cite{Tecchio2019NDimensional}:
\begin{equation}
	\begin{aligned}
		& D(\mathbf{p}_1,\cdots,\mathbf{p}_k;  \mathbf{q}_1,\cdots,\mathbf{q}_k)=  \\
		& 2(-\frac{1}{2})^k   \begin{vmatrix}
			0 		& 1 		 & 1 		  & \cdots & 1			\\
			1 		& D(\mathbf{p}_1,\mathbf{q}_1) & D(\mathbf{p}_1,\mathbf{q}_2) & \cdots & D(\mathbf{p}_1,\mathbf{q}_k) \\
			1 		& D(\mathbf{p}_2,\mathbf{q}_1) & D(\mathbf{p}_2,\mathbf{q}_2) & \cdots & D(\mathbf{p}_2,\mathbf{q}_k)	\\
			\vdots	&\vdots 	 & \vdots	  & \ddots & \vdots 	\\
			1 		& D(\mathbf{p}_k,\mathbf{q}_1) & D(\mathbf{p}_k,\mathbf{q}_2) & \cdots & D(\mathbf{p}_k,\mathbf{q}_k)
		\end{vmatrix}, 
	\end{aligned}	
	\label{equ:bideterminant}
\end{equation}
where $D(\mathbf{p}_i,\mathbf{q}_i)$ is the squared distance $d_{ij}^2$ between the nodes $\mathbf{p}_i$ and $\mathbf{q}_i$. Then, the linear representation of all node locations can form a linear system:
\begin{equation}
	\label{equ:linear_representation}
	\left[ {\begin{array}{*{20}{c}}
			{{{\mathbf{P}}_{\mathcal{A}}}} \\ 
			{{{{\mathbf{P}}}_{\mathcal{F}}}} 
	\end{array}} \right] = \left[ {\begin{array}{*{20}{c}}
			{\mathbf{I}}&{\mathbf{0}} \\ 
			{\mathbf{B}}&{\mathbf{C}} 
	\end{array}} \right]\left[ {\begin{array}{*{20}{c}}
			{{{\mathbf{P}}_{\mathcal{A}}}} \\ 
			{{{{\mathbf{P}}}_{\mathcal{F}}}} 
	\end{array}} \right].
\end{equation}
$\mathbf{A}=\left[ {\begin{array}{*{20}{c}}
		{\mathbf{I}}&{\mathbf{0}} \\ 
		{\mathbf{B}}&{\mathbf{C}} 
\end{array}} \right] \in\mathbb{R}^{(m+n)\times(m+n)}$ is the barycentric coordinates calculated by inter-node distances. Finally, the BLL problem is formulated as the following linear model:
\begin{equation}
	\label{equ:bll_problem}
	(\mathbf{I}-\mathbf{C})\mathbf{P}_{\mathcal{S}} = \mathbf{B}\mathbf{P}_{\mathcal{A}}.
\end{equation}
Thus, the general process of BLL is constructing the BLL model in (\ref{equ:bll_problem}) and solving $\mathbf{P}_{\mathcal{F}}$ from the model\cite{DILOC, khan2008distributed, ECHO,Tecchio2019NDimensional}. Note that the same construction manner also works in $\mathbb{R}^3$.  We then review the known network localizability conditions for NLL and BLL. 

\begin{table*}[ht]	
	\centering
	\caption{An overview of existing work about localizability. $\checkmark$:solved, $\times$: unsolved,  $-$: unreported.} 
	\label{table1}
	\begin{tabular}{|c|c|c|c|c|c|c|}
		\hline
		\multirow{2}{*}{} & \multicolumn{2}{c|}{Network Localizability}                                                         & \multicolumn{4}{c|}{Node Localizability}                                                           \\ \cline{2-7} 
		& \begin{tabular}[c]{@{}c@{}}Necessary and\\ Sufficient Condition\end{tabular} & \begin{tabular}[c]{@{}c@{}}Testing\\ Algorithm\end{tabular} & \begin{tabular}[c]{@{}c@{}}Necessary and\\ Sufficient Condition\end{tabular} & \begin{tabular}[c]{@{}c@{}}Necessary\\ Condition\end{tabular} & \begin{tabular}[c]{@{}c@{}}Sufficient\\ Condition\end{tabular} & \begin{tabular}[c]{@{}c@{}}Detection\\ Algorithm\end{tabular} \\ \hline
		NLL & $\checkmark$ & $\checkmark$  & $\times$  & $\checkmark$ & $\checkmark$ &   $\checkmark$                                                          \\ \hline
		BLL & $\checkmark$ & $-$  & $-$  & $-$ & $-$ &   $-$                                                             \\ \hline
	\end{tabular}
\end{table*}
\subsubsection{Network Localizability Conditions} 

In NLL, the network localizability problem is closely related to the graph rigidity\cite{eren2004rigidity,eren2016graph}. 
\begin{lemma}[NLL Network Localizability Condition\cite{eren2004rigidity}]
	A network $\mathcal{G}$ is localizable in 2D, if and only if it contains at least 3 anchors and it is global rigid (i.e., 3-connected and redundantly rigid).
	\label{lemma:loc_NLL}
\end{lemma} 


In BLL, the network localizability condition is based on a \emph{generated graph} $\mathcal{G^A}$ associated with the barycentric representation, instead of $\mathcal{G}$.
\begin{definition}[Generated Graph, $\mathcal{G^A}$]
	Given $\mathcal{G=(V,E)}$ and the constructed $\mathbf{A}$ in (\ref{equ:bll_problem}), the generated graph $\mathcal{G^A=(V, E^A)}$ of $\mathcal{G}$ is defined as a graph with the same node set $\mathcal{V}$. An edge $(i,j)\in \mathcal{E^A}$ if $\mathbf{A}_{ij}\neq 0$. 
\end{definition}
Then, the BLL network localizability condition is designed based on $\mathcal{G^A}$\cite{ECHO}. 
\begin{lemma}[BLL Network Localizability Condition\cite{ECHO}]
	Using BLL algorithms, all free nodes in a generic graph $\mathcal {G}$ are localizable in 2D if every free node can find at least three node-disjoint paths to anchors through only edges in $\mathcal{G^A}$.
	\label{lemma:loc_BLL}
\end{lemma}

So network localizability for NLL can be checked based on the graph property of $\mathcal {G}$, while network localizability for BLL can be checked based on the property of $\mathcal{G^A}$.  We then review the node localizability conditions. 

\subsubsection{Node Localizability Conditions} 

To judge whether a specific node is localizable, the necessary and sufficient condition has not been derived yet for either NLL or BLL. But there exist several sufficient or necessary conditions. 

In NLL, Goldenberg et al. \cite{goldenberg2005network} proposed a necessary condition called \emph{3P condition}, i.e., a localizable node must have three node-disjoint paths to three distinct anchors. Yang et al. also derived a sufficient condition called \emph{RRT-3B condition}, i.e., a localizable node needs to be in a redundantly rigid component that is 3-connected and has three anchors. Then, Yang et al. \cite{yang2011understanding} proposed a tightened necessary condition called \emph{RR-3P} that a localizable node must be in a redundant rigid component in addition to satisfying the 3P condition. They also gave a weaker sufficient condition called \emph{RR3P condition} that the nodes belong to a redundantly rigid component containing at least three anchors, where there are three node-disjoint paths connecting the free nodes and three anchors. A detailed discussion about network localizability and node localizability can be seen in \cite{zhang2012theoretic}.

In BLL, the issue of node localizability has not been reported. The network localizability and node localizability detection algorithms are then reviewed.

\subsubsection{Network Localizability and Node Localizability Detection Algorithms}


To test whether a network is localizable using NLL, algorithms have been designed to verify global rigidity by integrating the detection of redundant rigidity \cite{laman1970graphs} and the detection of 3-connectivity\cite{hopcroft1973dividing,milkri1988new}.  

To detect localizable nodes in NLL, there are mainly three representative algorithms: 

(1) The most representative method is \textbf{trilateration protocol (TP)}\cite{eren2004rigidity}. In $\mathbb{R}^d$, TP first marks the anchors as localizable and free nodes as unlocalizable. Then, each unlocalizable node continually checks its neighbors and changes its state to localizable if it finds at least $d+1$ localizable neighbors. 
TP is simple and efficient, but TP misses localizable nodes that are on a \emph{geographical gap} or a \emph{geographical border}. Moreover, TP fails to startup if the anchors do not have common neighbors; 

(2) To identify nodes on a geographical gap or a border, Yang et al.\cite{Yang2009} proposed the \textbf{wheel extension (WE)} algorithm. The structure of the wheel graph is defined and proved to be global rigid. Then each node tries to find a wheel graph in neighbors. If three or more nodes of a wheel are localizable, all other nodes of the wheel are localizable. WE inherits the efficiency of TP, but it may fail if there are not at least three anchors coexist in one wheel at the beginning, then no node can be localized; 

(3) To supplement TP and WE, Wu et al.\cite{wu2017triangle} defines the concept of branch, and proves that a branch including three anchors is localizable. A \textbf{triangle extension (TE)} method is proposed to construct branches. The branch construction starts from two connected anchors and extends the branch by extension operations. Once the extension reaches the third anchor or a localizable node, all nodes on the branch are localizable. However, TE fails when every two anchors do not share any common neighbor; Localizability detection in 3D has been explored in \cite{xia2016localizability}.

In BLL, there is neither a network localizability testing algorithm nor a node localizability detecting algorithm yet.



\subsection{Summary of BLL Localizability Detection and the Overview of Our Work}
An overview of the state-of-the-art for network localizability, node localizability detection conditions and  algorithms is summarized in Table~\ref{table1}. Overall, the localizability problem has been well researched in NLL except for a necessary and sufficient node localizability condition and the corresponding detection algorithm that guarantee to detect all localizable nodes. The research on localizability in BLL is much lacked. Only the necessary and sufficient condition for BLL network localizability is given in \cite{ECHO}. The testing algorithm for BLL network localizability, conditions and algorithm for BLL node localizability are still open.

To fill the knowledge gap in BLL, this paper proposes: 

1) A necessary condition for BLL node localizability. 

2) A sufficient condition for BLL node localizability. 

3) A testing algorithm for BLL network localizability, and 

4) A detection algorithm for BLL localizable nodes. 

The necessary and sufficient condition for node localizability is still left as an open problem.

\section{A Necessary Condition and A Sufficient Condition for BLL Node Localizability}

\label{sect:conditions}
Since the \emph{network localizability} condition for BLL has been given \cite{ECHO}, we will explore the \emph{node localizability} conditions for BLL in this section. For a network $\mathcal{G=\{V,E,\bf{d}\}}$, Yang et. al \cite{yang2011understanding} proposed that even $(v_i,v_j)\notin \mathcal{E}$, there might be an \emph{implicit edge} between the two nodes if $d_{ij}$ can be uniquely determined by two rigid constraints. Hereinafter, we assume that $\mathcal{G}$ is the graph after adding the implicit edges.

\subsection{Necessary Condition for BLL Node Localizability}
From Lemma~\ref{lemma:loc_BLL}, a node should have three node-disjoint paths (3P) to anchors in $\mathcal{G^A}$ to be BLL-localizable. To make a certain node has 3P to anchors, it is straightforward that it must have at least 3 neighbors in $\mathcal{G^A}$. From the perspective of the original graph $\mathcal{G}$, this necessary condition is formulated as follows. 
\begin{theorem}[Necessity of $d+1$ Mutually Connected Neighbors]
	\label{theorem:Necessity_BLL}
	In a graph $\mathcal{G=\{V,E\}}$ in $\mathbb{R}^d$, for a node to be localizable using BLL, it must have at least $d+1$ mutually connected neighbors. 
\end{theorem}
\begin{proof}
	Recall that in the construction of the BLL linear model, the barycentric coordinates are calculated using the Cayley-Menger bideterminant, where the distance between each pair of nodes from the input $\{v_1,\cdots,v_k\}$ is involved. Thus, for node $v_i$, $\exists\mathbf{A}_{ij}\neq 0$ (i.e., $v_j$ is a neighbor of $v_i$ in $\mathcal{G^A}$) only if $v_i$ has $d+1$ mutually connected neighbors so that the Cayley-Menger bideterminant can be calculated.
\end{proof}
Theorem~\ref{theorem:Necessity_BLL} implies that BLL node localizability requires stronger connectivity for each node. Neighbors in $\mathcal{N}_i$ need to be mutually connected to contribute to the BLL-localizability of $v_i$. That is, BLL-localizability is theoretically more difficult to be satisfied than NLL-localizability. For example, Yang et. al. have proved that a wheel graph like Fig.~\ref{fig:demo_bll_hardness} is global rigid (i.e., NLL-localizable)\cite{Yang2009}. But none of these nodes are BLL-localizable since each node's neighbors are not mutually connected. 
\begin{figure}[ht!]
	\centering
	\subfigure[A wheel graph.]{
		\begin{minipage}[t]{0.46\linewidth}
			\centering
			\includegraphics[width=.8\textwidth]{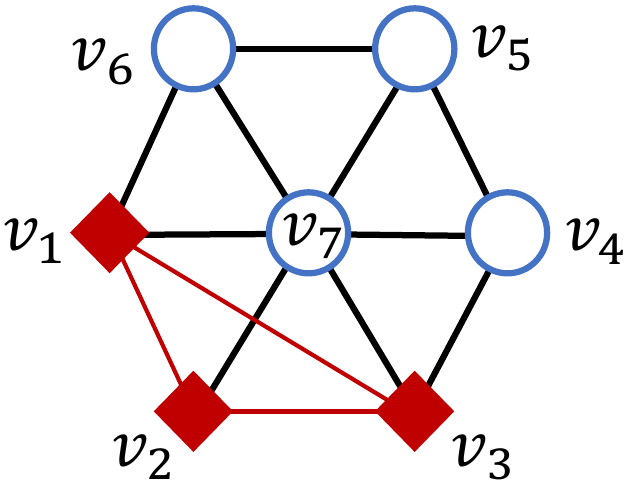}
			\label{fig:demo_bll_hardness}
		\end{minipage}
	}
	\subfigure[The corresponding $\mathcal{G^A}$.]{
		\begin{minipage}[t]{0.46\linewidth}
			\centering
			\includegraphics[width=0.8\textwidth]{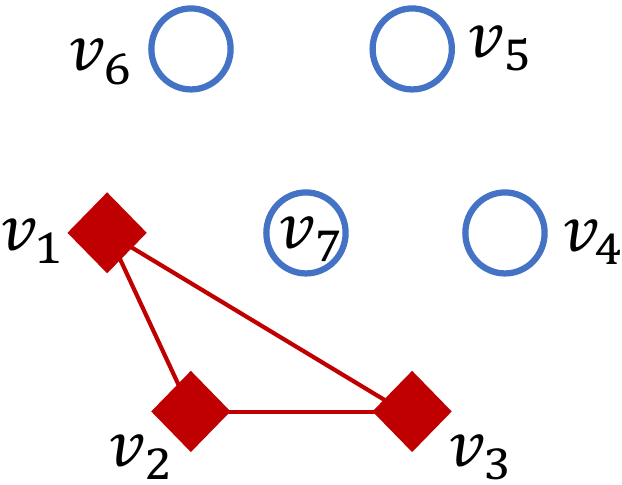}
			\label{fig:demo_bll_hardness_2}
		\end{minipage}
	}
	\caption{A wheel graph is NLL-localizable but not BLL-localizable.}
	\label{fig:bll_hardness}
\end{figure}
\subsection{Sufficient Condition for BLL Node Localizability}

From Lemma~\ref{lemma:loc_BLL}, we derive a sufficient condition for BLL node localizability by recursion.

\begin{theorem}[Sufficiency of  Recursive-3DP]
	\label{theorem:sufficient_BLL}
	In a graph $\mathcal{G=\{V,E\}}$ in $\mathbb{R}^2$, if 1) $v_i$ has 3P to anchors in $\mathcal{G^A}$, without loss of generality, say $p_1$, $p_2$, and $p_3$, and 2) every node on each path $p_1$, $p_2$, and $p_3$ has 3P to anchors. Then, node $v_i$ is BLL-localizable in $\mathcal{G}$. We call this condition  Recursive-3DP  for brief. 
\end{theorem}
\begin{proof}	
	If both 1) and 2) are satisfied, the graph induced by $v_i$ and the nodes it passes by on its paths satisfies Lemma~\ref{lemma:loc_BLL}. So all these nodes are BLL-localizable. 
	
	
\end{proof}
Theorem~\ref{theorem:sufficient_BLL} implies that the set of BLL-localizable nodes and the set of BLL-unlocalizable nodes should be edge-independent, i.e., the BLL-localizable nodes should have no edge connecting BLL-unlocalizable nodes. Thus, the BLL-localizable nodes can be found by iteratively removing nodes without 3P.





\section{Algorithms for BLL Network Localizability and BLL Node Localizability}
\label{sect:algorithms}
Given a network $\mathcal{G}$, there is no existing algorithm that can tell the network localizability or node localizability in BLL. From the previous section, we can observe that the key issue of BLL-localizability is verifying the $3P$ condition. Thus in this section, we first transform the verification of $3P$ to the calculation of max-flow in an appropriately constructed flow graph. Then, a max-flow-based algorithm is presented to test whether $\mathcal{G}$ is BLL-localizable by checking the max-flow in its $\mathcal{G^A}$. Furthermore, we adopt this idea to detect the BLL-localizable nodes using Theorem~\ref{theorem:sufficient_BLL} through an iterative algorithm. 

\subsection{Disjoint Path and Max-Flow}
\label{sect:dp_max_flow}
To proceed, we introduce a flow graph $\mathcal{G^F}=(\mathcal{V^F},\mathcal{E^F},\mathcal{C})$ constructed from $\mathcal{G=\{V,E\}}$, where $\mathcal{V}=\mathcal{A}\cup\mathcal{F}$.
\begin{itemize}
	\item Divide each free node $v_i\in\mathcal{F}$ into two copies $\{v_i^{in}, v_i^{out}\}$ and add them to $\mathcal{V^F}$. Then, add anchor nodes in $\mathcal{A}$ and an virtual target node $\Omega$ to $\mathcal{V^F}$.
	\item  For each $v_i\in\mathcal{F}$, add $(v_i^{in}, v_i^{out})$ to $\mathcal{E^F}$. For each neighbor $v_j\in\mathcal{N}_i$, add $(v_i^{out},v_j)$ to $\mathcal{E^F}$ if $v_j$ is an anchor, add $(v_i^{out},v_j^{in})$ to $\mathcal{E^F}$, otherwise; For each anchor node $v_i\in\mathcal{A}$, add $(v_i,\Omega)$ to $\mathcal{E^F}$. 
	\item Assign $\mathcal{C}_{ij}=1$ in the capacity matrix $\mathcal{C}$ if $(v_i,v_j)\in\mathcal{E^F}$.
\end{itemize}

Then, we show the equality between the number of disjoint paths (DP) in $\mathcal{G}$ and the Max-Flow in $\mathcal{G^F}$.

\begin{lemma}
	Consider a network $\mathcal{G}$ and its corresponding $\mathcal{G^F}$, the following two are equivalent.
	\begin{itemize}
		\item The count of disjoint paths from $v_i$ to anchors in $\mathcal{G}$.
		\item The max flow from $v_i^{out}$ to $\Omega$ in $\mathcal{G^F}$.
	\end{itemize}
	\label{lemma:max_flow_dp}
\end{lemma}
\begin{proof}
	First we prove that every path to anchors in $\mathcal{G}$ is included in $\mathcal{G^F}$. Suppose an arbitrary edge $(v_i,v_j)$ on any path $\mathcal{P}=v_1\cdots v_k$ in $\mathcal{G}$. If $v_i$ is an free node and $v_j$ is an anchor node, then the edge is maintained by $v_i^{in}\rightarrow v_i^{out}\rightarrow v_j$. If both $v_i$ and $v_j$ are free nodes, the edge is maintained by $v_i^{in}\rightarrow v_i^{out}\rightarrow v_j^{in}\rightarrow v_j^{out}$. Since we focus on paths to anchors, the case that $v_i$ is an anchor node and $v_j$ is a free node is not considered. Thus, any path $\mathcal{P}$ from $v_i$ to anchors in $\mathcal{G}$ can be transformed to a new path in $\mathcal{G^F}$; 	
	Suppose any source node $v_i^{out}$, it is straightforward that any path from $v_i^{out}$ to $\Omega$ must go through one edge of the minimum cut. And each node $v_j$ is passed by at most one path since the capacity of $(v_j^{in},v_j^{out})$ is one. Because the maximum flow equals the minimum edge cut, the count of disjoint paths from $v_i$ to anchors in $\mathcal{G}$ equals the max flow from $v_i^{out}$ to $\Omega$ in $\mathcal{G^F}$.
\end{proof}

A case study is shown in Fig.~\ref{fig:demo_max_flow}. The generated graph $\mathcal{G^A}$ of a network $\mathcal{G}$ of 5 nodes is given in Figure \ref{fig:demo_flow_G}. Its corresponding $\mathcal{G^F}$ is in Figure \ref{fig:demo_flow_G_A}. The constructed capacity matrix $\mathcal{C}$ is shown in Table~\ref{tab:capacity}. 
It is easy to verify the equivalence between the number of disjoint paths in $\mathcal{G^A}$ and the max flow in $\mathcal{G^F}$. For example, there are five paths from $v_4$ to anchors, i.e.,  $v_4\rightarrow v_1, v_4\rightarrow v_5 \rightarrow v_1, v_4\rightarrow v_5 \rightarrow v_2, v_4\rightarrow v_5 \rightarrow v_3, v_4\rightarrow v_6 \rightarrow v_3$. Three of them are node disjoint, i.e., $v_4\rightarrow v_1, v_4\rightarrow v_5 \rightarrow v_2, v_4\rightarrow v_6 \rightarrow v_3$. Using Orlin’s method \cite{orlin2013max}, we can obtain that the max-flow from $v_4$ to $\Omega$ in $\mathcal{G^F}$ is also three.



\begin{figure}[ht!]
	\centering
	\subfigure[$\mathcal{G^A}$]{
		\begin{minipage}[t]{0.42\linewidth}
			\centering
			\includegraphics[width=.75\textwidth]{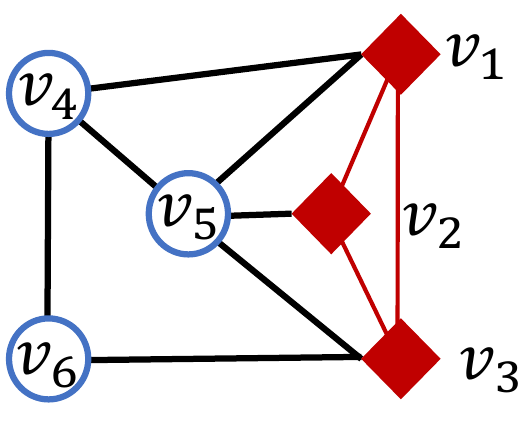}
			\label{fig:demo_flow_G}
		\end{minipage}
	}
	\subfigure[$\mathcal{G}^{F}$]{
		\begin{minipage}[t]{0.46\linewidth}
			\centering
			\includegraphics[width=0.99\textwidth]{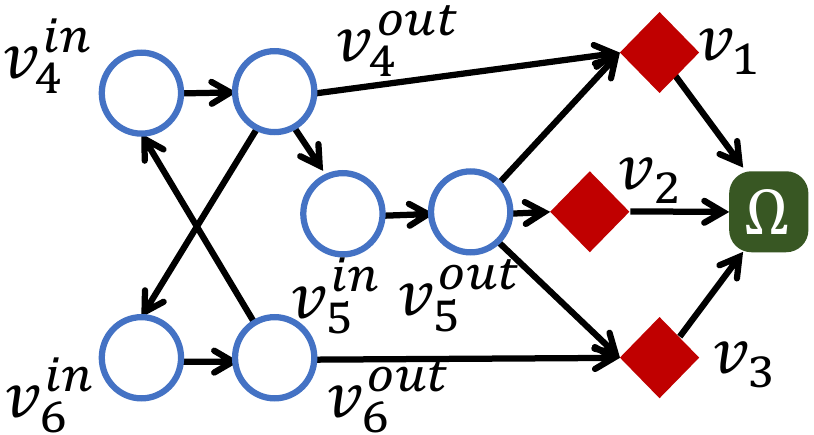}
			\label{fig:demo_flow_G_A}
		\end{minipage}
	}
	\caption{A demonstration of the construction of $\mathcal{G^F}$.}
	\label{fig:demo_max_flow}
\end{figure}

\newcommand\cellwidth{0.08cm}
\begin{table}[ht]
	\footnotesize
	\caption{The capacity matrix $\mathcal{C}$ of the flow graph $\mathcal{G^F}$ in Fig.~\ref{fig:demo_flow_G_A}.}
	\centering
	\begin{tabular}{p{\cellwidth}p{\cellwidth}p{\cellwidth}p{\cellwidth}p{\cellwidth}p{\cellwidth}p{\cellwidth}p{\cellwidth}p{\cellwidth}p{\cellwidth}p{\cellwidth}p{\cellwidth}}		
		& $v_4^{in}$                                             & $v_4^{out}$                                            & $v_5^{in}$                                             & $v_5^{out}$                                           & $v_6^{in}$                                            & $v_6^{out}$                                            & 1                                              & 2                                              & 3                                              & $\Omega$                                             &       \\ \cline{1-11} 
		\multicolumn{1}{|l|}{\cellcolor[HTML]{ECF4FF}}                               & \multicolumn{1}{l|}{\cellcolor[HTML]{ECF4FF}0} & \multicolumn{1}{l|}{\cellcolor[HTML]{ECF4FF}1} & \multicolumn{1}{l|}{\cellcolor[HTML]{ECF4FF}0} & \multicolumn{1}{l|}{\cellcolor[HTML]{ECF4FF}0} & \multicolumn{1}{l|}{\cellcolor[HTML]{ECF4FF}0} & \multicolumn{1}{l|}{\cellcolor[HTML]{ECF4FF}0} & \multicolumn{1}{l|}{\cellcolor[HTML]{ECF4FF}0} & \multicolumn{1}{l|}{\cellcolor[HTML]{ECF4FF}0} & \multicolumn{1}{l|}{\cellcolor[HTML]{ECF4FF}0} & \multicolumn{1}{l|}{\cellcolor[HTML]{ECF4FF}0} & $v_4^{in}$     \\ \cline{2-11}
		\multicolumn{1}{|l|}{\cellcolor[HTML]{ECF4FF}}                               & \multicolumn{1}{l|}{\cellcolor[HTML]{ECF4FF}0} & \multicolumn{1}{l|}{\cellcolor[HTML]{ECF4FF}0} & \multicolumn{1}{l|}{\cellcolor[HTML]{ECF4FF}1} & \multicolumn{1}{l|}{\cellcolor[HTML]{ECF4FF}0} & \multicolumn{1}{l|}{\cellcolor[HTML]{ECF4FF}1} & \multicolumn{1}{l|}{\cellcolor[HTML]{ECF4FF}0} & \multicolumn{1}{l|}{\cellcolor[HTML]{ECF4FF}1} & \multicolumn{1}{l|}{\cellcolor[HTML]{ECF4FF}0} & \multicolumn{1}{l|}{\cellcolor[HTML]{ECF4FF}0} & \multicolumn{1}{l|}{\cellcolor[HTML]{ECF4FF}0} & $v_4^{out}$     \\ \cline{2-11}
		\multicolumn{1}{|l|}{\cellcolor[HTML]{ECF4FF}}                               & \multicolumn{1}{l|}{\cellcolor[HTML]{ECF4FF}0} & \multicolumn{1}{l|}{\cellcolor[HTML]{ECF4FF}0} & \multicolumn{1}{l|}{\cellcolor[HTML]{ECF4FF}0} & \multicolumn{1}{l|}{\cellcolor[HTML]{ECF4FF}1} & \multicolumn{1}{l|}{\cellcolor[HTML]{ECF4FF}0} & \multicolumn{1}{l|}{\cellcolor[HTML]{ECF4FF}0} & \multicolumn{1}{l|}{\cellcolor[HTML]{ECF4FF}0} & \multicolumn{1}{l|}{\cellcolor[HTML]{ECF4FF}0} & \multicolumn{1}{l|}{\cellcolor[HTML]{ECF4FF}0} & \multicolumn{1}{l|}{\cellcolor[HTML]{ECF4FF}0} & $v_5^{in}$     \\ \cline{2-11}
		\multicolumn{1}{|l|}{\cellcolor[HTML]{ECF4FF}}                               & \multicolumn{1}{l|}{\cellcolor[HTML]{ECF4FF}1} & \multicolumn{1}{l|}{\cellcolor[HTML]{ECF4FF}0} & \multicolumn{1}{l|}{\cellcolor[HTML]{ECF4FF}0} & \multicolumn{1}{l|}{\cellcolor[HTML]{ECF4FF}0} & \multicolumn{1}{l|}{\cellcolor[HTML]{ECF4FF}0} & \multicolumn{1}{l|}{\cellcolor[HTML]{ECF4FF}0} & \multicolumn{1}{l|}{\cellcolor[HTML]{ECF4FF}1} & \multicolumn{1}{l|}{\cellcolor[HTML]{ECF4FF}1} & \multicolumn{1}{l|}{\cellcolor[HTML]{ECF4FF}1} & \multicolumn{1}{l|}{\cellcolor[HTML]{ECF4FF}0} & $v_5^{out}$     \\ \cline{2-11}
		\multicolumn{1}{|l|}{\cellcolor[HTML]{ECF4FF}}                               & \multicolumn{1}{l|}{\cellcolor[HTML]{ECF4FF}0} & \multicolumn{1}{l|}{\cellcolor[HTML]{ECF4FF}0} & \multicolumn{1}{l|}{\cellcolor[HTML]{ECF4FF}0} & \multicolumn{1}{l|}{\cellcolor[HTML]{ECF4FF}0} & \multicolumn{1}{l|}{\cellcolor[HTML]{ECF4FF}0} & \multicolumn{1}{l|}{\cellcolor[HTML]{ECF4FF}1} & \multicolumn{1}{l|}{\cellcolor[HTML]{ECF4FF}0} & \multicolumn{1}{l|}{\cellcolor[HTML]{ECF4FF}0} & \multicolumn{1}{l|}{\cellcolor[HTML]{ECF4FF}0} & \multicolumn{1}{l|}{\cellcolor[HTML]{ECF4FF}0} & $v_6^{in}$     \\ \cline{2-11}
		\multicolumn{1}{|l|}{\multirow{-6}{*}{\cellcolor[HTML]{ECF4FF}\rotatebox{90}{Free Nodes}}}   & \multicolumn{1}{l|}{\cellcolor[HTML]{ECF4FF}1} & \multicolumn{1}{l|}{\cellcolor[HTML]{ECF4FF}0} & \multicolumn{1}{l|}{\cellcolor[HTML]{ECF4FF}0} & \multicolumn{1}{l|}{\cellcolor[HTML]{ECF4FF}0} & \multicolumn{1}{l|}{\cellcolor[HTML]{ECF4FF}0} & \multicolumn{1}{l|}{\cellcolor[HTML]{ECF4FF}0} & \multicolumn{1}{l|}{\cellcolor[HTML]{ECF4FF}0} & \multicolumn{1}{l|}{\cellcolor[HTML]{ECF4FF}0} & \multicolumn{1}{l|}{\cellcolor[HTML]{ECF4FF}1} & \multicolumn{1}{l|}{\cellcolor[HTML]{ECF4FF}0} & $v_6^{out}$     \\ \cline{1-11}
		\multicolumn{1}{|l|}{\cellcolor[HTML]{FFCCC9}}                               & \multicolumn{1}{l|}{\cellcolor[HTML]{FFCCC9}0} & \multicolumn{1}{l|}{\cellcolor[HTML]{FFCCC9}0} & \multicolumn{1}{l|}{\cellcolor[HTML]{FFCCC9}0} & \multicolumn{1}{l|}{\cellcolor[HTML]{FFCCC9}0} & \multicolumn{1}{l|}{\cellcolor[HTML]{FFCCC9}0} & \multicolumn{1}{l|}{\cellcolor[HTML]{FFCCC9}0} & \multicolumn{1}{l|}{\cellcolor[HTML]{FFCCC9}0} & \multicolumn{1}{l|}{\cellcolor[HTML]{FFCCC9}0} & \multicolumn{1}{l|}{\cellcolor[HTML]{FFCCC9}0} & \multicolumn{1}{l|}{\cellcolor[HTML]{FFCCC9}1} & 1     \\ \cline{2-11}
		\multicolumn{1}{|l|}{\cellcolor[HTML]{FFCCC9}}                               & \multicolumn{1}{l|}{\cellcolor[HTML]{FFCCC9}0} & \multicolumn{1}{l|}{\cellcolor[HTML]{FFCCC9}0} & \multicolumn{1}{l|}{\cellcolor[HTML]{FFCCC9}0} & \multicolumn{1}{l|}{\cellcolor[HTML]{FFCCC9}0} & \multicolumn{1}{l|}{\cellcolor[HTML]{FFCCC9}0} & \multicolumn{1}{l|}{\cellcolor[HTML]{FFCCC9}0} & \multicolumn{1}{l|}{\cellcolor[HTML]{FFCCC9}0} & \multicolumn{1}{l|}{\cellcolor[HTML]{FFCCC9}0} & \multicolumn{1}{l|}{\cellcolor[HTML]{FFCCC9}0} & \multicolumn{1}{l|}{\cellcolor[HTML]{FFCCC9}1} & 2     \\ \cline{2-11}
		\multicolumn{1}{|l|}{\multirow{-3}{*}{\cellcolor[HTML]{FFCCC9}\rotatebox{90}{Anchors}}} & \multicolumn{1}{l|}{\cellcolor[HTML]{FFCCC9}0} & \multicolumn{1}{l|}{\cellcolor[HTML]{FFCCC9}0} & \multicolumn{1}{l|}{\cellcolor[HTML]{FFCCC9}0} & \multicolumn{1}{l|}{\cellcolor[HTML]{FFCCC9}0} & \multicolumn{1}{l|}{\cellcolor[HTML]{FFCCC9}0} & \multicolumn{1}{l|}{\cellcolor[HTML]{FFCCC9}0} & \multicolumn{1}{l|}{\cellcolor[HTML]{FFCCC9}0} & \multicolumn{1}{l|}{\cellcolor[HTML]{FFCCC9}0} & \multicolumn{1}{l|}{\cellcolor[HTML]{FFCCC9}0} & \multicolumn{1}{l|}{\cellcolor[HTML]{FFCCC9}1} & 3     \\ \cline{1-11}
		\multicolumn{1}{|l|}{\cellcolor[HTML]{009901}$\Omega$}                          & \multicolumn{1}{l|}{\cellcolor[HTML]{009901}0} & \multicolumn{1}{l|}{\cellcolor[HTML]{009901}0} & \multicolumn{1}{l|}{\cellcolor[HTML]{009901}0} & \multicolumn{1}{l|}{\cellcolor[HTML]{009901}0} & \multicolumn{1}{l|}{\cellcolor[HTML]{009901}0} & \multicolumn{1}{l|}{\cellcolor[HTML]{009901}0} & \multicolumn{1}{l|}{\cellcolor[HTML]{009901}0} & \multicolumn{1}{l|}{\cellcolor[HTML]{009901}0} & \multicolumn{1}{l|}{\cellcolor[HTML]{009901}0} & \multicolumn{1}{l|}{\cellcolor[HTML]{009901}0} & $\Omega$ \\ \cline{1-11}
	\end{tabular}
	\label{tab:capacity}
\end{table}

\subsection{BLL Network Localizability Test}
Now, whether a network $\mathcal{G}$ is BLL-localizable can be tested by checking $\mathcal{G^A}$ using Lemma~\ref{lemma:max_flow_dp}. Algorithm~\ref{alg:BLL_Test} shows the routine. Recall that a network is BLL-localizable means that all the nodes are BLL-localizable. So if any node does not meet the necessary condition in Theorem~\ref{theorem:Necessity_BLL}, the network is not BLL-localizable (Line~\ref{mf:necessary_start}-Line~\ref{mf:necessary_end}). Then, if all nodes meet the necessary condition, we construct the generated graph $\mathcal{G^A}$ shown as \textbf{Function} \emph{construct\_generated\_graph}($\mathcal{G}$) (Line~\ref{alg:construct_ga_begin}-Line~\ref{alg:construct_ga_end}). Then the flow graph $\mathcal{G^F}$ of $\mathcal{G^A}$ is constructed using the method in Section~\ref{sect:dp_max_flow}. The max flow is calculated using Orlin’s method \cite{orlin2013max} (Line~\ref{mf:cal_max_flow_mf}).
In Lemma~\ref{lemma:loc_BLL}, the BLL network localizability condition requires that each free node should have 3P to anchors in $\mathcal{G^A}$, so $\mathcal{G}$ is BLL-localizable if the max-flow from each $v_i\in\mathcal{F}$ to $\Omega$ is not less than 3 in $\mathcal{G^F}$, BLL-unlocalizable, otherwise (Line~\ref{alg:network_loc_start}-Line~\ref{alg:network_loc_end}).

Use the graph in Fig.~\ref{fig:demo_max_flow} for instance. Let $\bf{DP}^{\mathcal{G_A}}_{i\rightarrow \mathcal{A}}$ be the number of disjoint paths from $v_i$ to nodes in $\mathcal{A}$ in $\mathcal{G_A}$, and $\bf{Max\_Flow}_{i\rightarrow\Omega}^{\mathcal{G^F}}$ be the max flow from $v_i$ to the virtual node $\Omega$. We can obtain:

\begin{equation}
	\left\{
	\begin{aligned} 
		\mathbf{DP}_{4\rightarrow\{1,2,3\}}^{\mathcal{G^A}} & = \mathbf{Max\_Flow}_{4\rightarrow\Omega}^{\mathcal{G^F}}=3,\\
		\mathbf{DP}_{5\rightarrow\{1,2,3\}}^{\mathcal{G^A}} & = \mathbf{Max\_Flow}_{5\rightarrow\Omega}^{\mathcal{G^F}}=3, \\
		\mathbf{DP}_{6\rightarrow\{1,2,3\}}^{\mathcal{G^A}} & = \mathbf{Max\_Flow}_{6\rightarrow\Omega}^{\mathcal{G^F}}=2.
	\end{aligned}\right.
	\label{equ:demo_dp_mf}
\end{equation}
Thus, through Algorithm~\ref{alg:BLL_Test}, we can conclude that the network in Fig.~\ref{fig:demo_max_flow} is not BLL-localizable due to the existence of $v_6$.

\begin{algorithm}[htb!]
	\caption{\textbf{The Max-Flow (MF) Algorithm for BLL Network Localizability Test}}
	\label{alg:BLL_Test}
	\KwIn{$\mathcal{G=\{V,E\}}, \mathcal{V=A\cup F}$.}
	\KwOut{Network BLL-localizability: \emph{true} or \emph{false}.}
	\nl add the implicit edges to $\mathcal{E}$\cite{yang2011understanding}\;
	\lnl{mf:necessary_start} \If{$\exists v_i\in\mathcal{V}$ has no 3 mutually connected neighbors}
	{
		\tcp{\small{the necessary condition in Theorem~\ref{theorem:Necessity_BLL} is not satisfied}}
		\lnl{mf:necessary_end} \KwRet{false}\;
	}
	\nl $\mathcal{G^A}=\{\mathcal{V},\mathcal{E^A}\}\gets construct\_generated\_graph(\mathcal{G})$\;
	\nl construct the flow graph $\mathcal{G^F}=(\mathcal{V^F},\mathcal{E^F},\mathcal{C})$ of $\mathcal{G^A}$ as in Section~\ref{sect:dp_max_flow}\;
	\lnl{mf:cal_max_flow_mf} calculate  $\mathbf{Max\_Flow}_{i\rightarrow\Omega}^{\mathcal{G^F}}$ for each $v_i\in\mathcal{F}$\;
	\lnl{alg:network_loc_start} \eIf{$\exists$  $\mathbf{Max\_Flow}_{i\rightarrow\Omega}^{\mathcal{G^F}}<3$}
	{
		\nl \KwRet{false}\; 
	}
	{
		\lnl{alg:network_loc_end} \KwRet{true}\;
	}
	\Fn{construct\_generated\_graph($\mathcal{G}$)}
	{
		\lnl{alg:construct_ga_begin}\For{each node $v_i\in \mathcal{F}$}
		{
			\nl $tri\_count\gets 0$\;
			\nl\For{each combination of 3 mutually connected nodes $\{v_j,v_k,v_l\}$ in $\mathcal{N}_i$}
			{
				\nl $tri\_count\gets tri\_count+1$\;
				\nl $\{\mathbf{A}_{ij}^{tri\_count},\mathbf{A}_{ik}^{tri\_count},\mathbf{A}_{il}^{tri\_count}\}\gets$ (\ref{equ:barycentric_representation})\;
				\nl $\mathbf{A}_{is}^{tri\_count}\gets 0$ for $s\in \mathcal{N}_i\setminus\{v_j,v_k,v_l\}$\;
				
			}
			\nl $\mathbf{A}_{ir} = \frac{1}{tri\_count}\sum_{t=1}^{tri\_count} \mathbf{A}_{ir}^{t}$, for $r\in\mathcal{N}_i$
			\;
		}
		\nl $\mathcal{G^A}=\{\mathcal{V},\mathcal{E^A}\}$, where $(v_i,v_j)\in\mathcal{E^A}$ if $\mathbf{A}_{ij}\neq 0$\;
		\lnl{alg:construct_ga_end} \KwRet{$\mathcal{G^A}$}\;
	}	
\end{algorithm}

\subsection{BLL Node Localizability Detection}
Detecting the BLL-localizable nodes in a partial localizable network is conducted through Theorem~\ref{theorem:sufficient_BLL}, i.e., detecting the nodes that satisfy the  Recursive-3DP  condition. The routine is shown as the IMF algorithm in Algorithm~\ref{alg:BLL_Detection}. 
IMF iteratively removes nodes not having 3P and the corresponding capacities of these nodes in each round. In each round, only the nodes having 3P and the edges among them (denoted by \textbf{Edge}($\cdot$)) are left (Line~\ref{alg:remove_capacity}-Line~\ref{alg:go_back}). Finally, IMF terminates when no more nodes can be removed (Line~\ref{alg:ret}). 

Take Fig.~\ref{fig:demo_max_flow} as an example to show the process of IMF. In the first round, $v_6$ is removed since $\mathbf{Max\_Flow}_{6\rightarrow\Omega}^{\mathcal{G^F}}=2$, then   $\mathcal{F}^{*}=\{v_4,v_5\}$; In the second round,  $\mathbf{Max\_Flow}_{4\rightarrow\Omega}^{\mathcal{G^F}}$ becomes 2 since $v_6$ and its corresponding edges are removed. Then, $\mathcal{F}^{*}$ becomes $\{v_5\}$; In the third round, $\mathcal{F}^{*}$ does not change since the 3P of $v_5$ does not pass through nodes without 3P. So IMF terminates and outputs $\{v_5\}$ as the BLL-localizable node set.


\begin{algorithm}[htb!]
	\caption{\textbf{The Iterative Max-Flow (IMF) Algorithm for BLL Node Localizability Detection}}
	\label{alg:BLL_Detection}	
	\KwIn{$\mathcal{G=\{V,E\}}, \mathcal{V=A\cup F}$.}
	\KwOut{Set of BLL-localizable nodes: $\mathcal{F}^{*}$.}
	\nl add the implicit edges to $\mathcal{E}$\cite{yang2011understanding}\;
	\lnl{alg:construct_g_a} $\mathcal{G^A}=\{\mathcal{V},\mathcal{E^A}\}\gets construct\_generated\_graph(\mathcal{G})$\;
	\lnl{alg:construct_g_f} construct the flow graph $\mathcal{G^F}=(\mathcal{V^F},\mathcal{E^F},\mathcal{C})$ of $\mathcal{G^A}$ as in Section~\ref{sect:dp_max_flow}\;
	\lnl{alg:cal_max_flow_imf} calculate $\mathbf{Max\_Flow}_{i\rightarrow\Omega}^{\mathcal{G^F}}$ for each $v_i\in\mathcal{F}$\;
	\nl $\mathcal{F}^{*}\gets\{v_i|\mathbf{Max\_Flow}_{i\rightarrow\Omega}^{\mathcal{G^F}}\geq 3\}$; \tcp{\small{the nodes whose max flow is not less than 3}}
	\nl \eIf{$\mathcal{F}^{*} = \mathcal{F}$}
	{
		\lnl{alg:ret} \KwRet{$\mathcal{F}^{*}$}\;
	}
	{
		\lnl{alg:remove_capacity} remove the capacities of $\mathcal{F}\setminus\mathcal{F}^{*}$ from $\mathcal{C}$\;
		\nl $\mathcal{F}\gets\mathcal{F}^{*}$, 
		$\mathcal{G}\gets\{\mathcal{A}\cup\mathcal{F}, \mathbf{Edge}(\mathcal{A}\cup\mathcal{F})\}$\;
		\lnl{alg:update_g_a} reconstruct $\mathcal{G^A}$ of $\mathcal{G}$\;\lnl{alg:go_back} go back to Line~\ref{alg:construct_g_f}\;
	}	
\end{algorithm}


\section{Analysis and Discussion}
\label{sect:analysis}
In this section, we first give the theoretical analysis of our IMF algorithm. Then we discuss its extensions. Finally, we give an application of IMF in NL.

\subsection{Analysis of IMF}

\subsubsection{Validity}

\begin{theorem}
	Given $\mathcal{G=\{V,E\}}$, all the nodes satisfying Theorem~\ref{theorem:sufficient_BLL} can be detected by Algorithm~\ref{alg:BLL_Detection}.
	\label{lemma:validity_IMF}
\end{theorem}

\begin{proof}
	Index the nodes removed by Algorithm \ref{alg:BLL_Detection} as $\mathcal{F}\setminus\mathcal{F}^{*} = \{s_1\cdots s_c\}$. Let $\mathcal{V}^{'}$ be an arbitrary feasible solution of Lemma \ref{lemma:loc_BLL} and $\mathbf{G}[V^{'}]$ be the induced subgraph of $\mathcal{V}^{'}$. $s_1$ is removed because $s_1$ has not three disjoint paths to anchors. There must be that $s_1$ cannot find three disjoint paths in $\mathbf{G}[V^{'}]$ since $\mathbf{G}[V^{'}]$ is a subgraph of $\mathcal{G}$. i.e., $s_1$ will not be included by any other feasible solution. Similarly, every node of $\mathcal{F}\setminus\mathcal{F}^{*}$ is not included by any feasible solution. Additionally, every node of $\mathcal{F}^{*}$ satisfies Theorem~\ref{theorem:sufficient_BLL}, otherwise it will be removed. Thus, Algorithm \ref{alg:BLL_Detection} selects all the vertices that can satisfy Theorem~\ref{theorem:sufficient_BLL}.
\end{proof}

\subsubsection{Time Complexity}

\begin{theorem}
	For a graph $\mathcal{G=\{V,E\}}$, let $|\mathcal{V}|=n$, $|\mathcal{E}|=m$ and the number of neighbors bounded by $\Delta$, i.e., $|\mathcal{N}_i|\leq\Delta,~\forall v_i\in\mathcal{V}$. The complexity of the IMF algorithm is $O(kn(\Delta^3+m))$, where $k$ is the number of iterations.
	\label{lemma:time_complexity_IMF}
\end{theorem}
\begin{proof}
	In each round of IMF, the construction of $\mathcal{G^A}$ needs a complexity of $O(n\Delta^3)$ since it needs to check the combination of three neighbors for each node. Then, the construction of $\mathcal{G^F}$ checks each edge in $\mathcal{E}$ to assign the capacities, requesting a complexity of $O(m)$. In calculating the max flow, there exist many efficient algorithms\cite{goldberg2014efficient}. In this article, we adopt Orlin's method with a complexity of $O(nm)$\cite{orlin2013max}. Let $k$ be the iteration count, the complexity of IMF is $O(kn(\Delta^3+m))$. In the worst case, only one node is detected to be with a Max-Flow less than 3 in each round and then removed, then the time complexity becomes $O(n^2(\Delta^3+m))$.
\end{proof}

\subsubsection{Space Complexity}
The major space consumption of IMF is storing $\mathcal{G^A}$ and $\mathcal{G^F}$, and the space complexity is $O(n^2)$.

\subsection{Extensions of IMF}

\subsubsection{Extension to NLL Localizability Detection}
Recall that IMF checks the  Recursive-3DP  condition in the generated graph $\mathcal{G_A}$. If we modify IMF as follows, the returned $\mathcal{F}^{*}$ becomes the set of nodes satisfying Recursive-3DP in the original graph $\mathcal{G}$.

\textbf{Modification 1 (M1):} 
\begin{itemize}
	\item Remove Line~\ref{alg:construct_g_a} and Line~\ref{alg:update_g_a} from Algorithm~\ref{alg:BLL_Detection}. 
	\item Change the construction of $\mathcal{G^F}$ in Line~\ref{alg:construct_g_f} to be based on $\mathcal{G}$.
\end{itemize}

Thus, the returned $\mathcal{F}^{*}$ becomes the set of nodes satisfying  Recursive-3DP  in $\mathcal{G}$ after \textbf{M1}.

\begin{theorem}
	A graph $\mathcal{G}$ is global rigid if it is composed of Recursive-3DP nodes.
	\label{lemma:IMF_NLL}
\end{theorem}
\begin{proof}
	If each node has three node disjoint paths to anchors, there is no binary vertex cut. Then $\mathcal{G}$ is redundant rigid and 3-connected so that the NLL localizability condition is met. 
\end{proof}
That is, IMF returns the NLL-localizable nodes if we conduct \textbf{M1}.

\subsubsection{Extension to Arbitrary Dimensions}

IMF works in any dimensions, denoted by $d$.

\textbf{Modification 2 (M2):} 
\begin{itemize}
	\item Change the threshold in Line~\ref{alg:cal_max_flow_imf} to $d+1$.
	\item Increase the number of anchors to be no less than $d+1$. 
\end{itemize}


Then, the returned nodes form a $(d+1)$P$^{+}$ graph, where each node has at least $d+1$ disjoint paths to anchors. Although the equivalence between globally rigidity and localizability has not been proved when $d>2$, IMF provides an efficient way to detect a well-structured $(d+1)$P$^{+}$ subgraph after \textbf{M2}.

\section{Evaluation}
\label{sect:evaluation}
Simulations are conducted using MATLAB R2020b to show the effectiveness of our method. First, we present the localizability detection capacity of the proposed IMF algorithm compared with the most representative localizability detection algorithm TP and the state-of-the-art algorithm TE. Then, we show the strong applicability of IMF in a variety of harsh situations. Finally, we analyze the application of IMF in BLL.



\subsection{Effectiveness of IMF in Localizability Detection}

In this section, the perception radius $R$ is varied to control the network density (assessed by average node degree). The percentage of anchor nodes (i.e., $|\mathcal{A}|/|\mathcal{V}|$) is also varied. The detection capacity of a certain algorithm is assessed by the percentage of detected localizable nodes $|\hat{\mathcal{F}}^*|/|\mathcal{V}|$ (called $\%$ \emph{of localizable} for brief), where $\hat{\mathcal{F}}^*$ is the detected localizable nodes by a certain algorithm. 

\begin{figure*}[t!]
	\centering
	\subfigure[TP]{
		\begin{minipage}[t]{0.3\linewidth}
			\centering
			\includegraphics[width=0.99\textwidth]{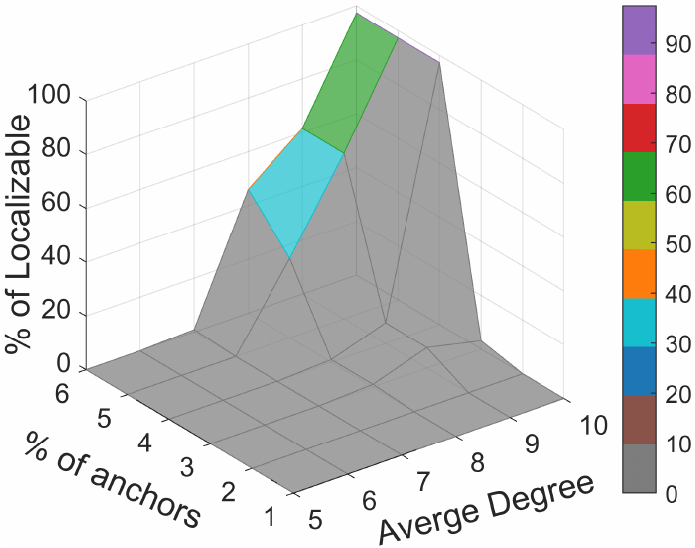}
			\label{fig:eva_tp}
		\end{minipage}
	}
	\subfigure[TE]{
		\begin{minipage}[t]{0.3\linewidth}
			\centering
			\includegraphics[width=.99\textwidth]{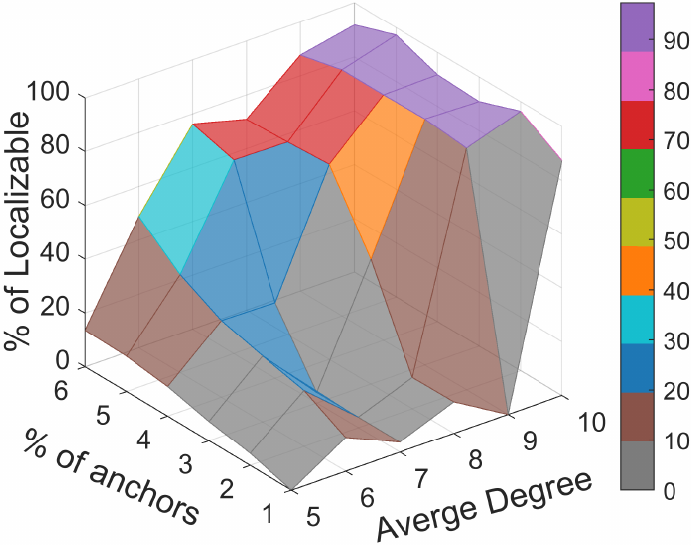}
			\label{fig:eva_te}
		\end{minipage}
	}
	\subfigure[IMF]{
		\begin{minipage}[t]{0.3\linewidth}
			\centering
			\includegraphics[width=0.99\textwidth]{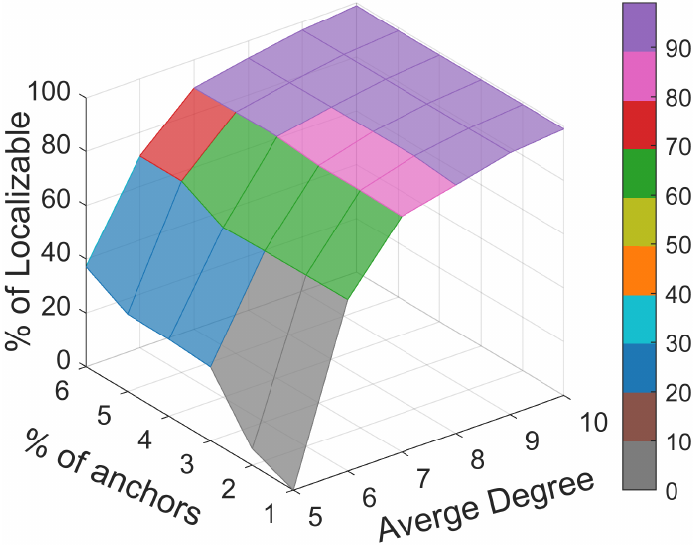}
			\label{fig:eva_imf}
		\end{minipage}
	}
	\caption{The effect of anchor density and node density on detection capacity.}
	\label{fig:eva_localizability}
\end{figure*}

Fig.~\ref{fig:eva_localizability} shows the percentage of detected localizable nodes of different algorithms under different average degrees and anchor densities. 
We can obtain the following observations:
\begin{itemize}
	\item IMF outperforms the other two algorithms in detection capacity under all settings, especially in sparse and anchor-lacking networks. 
	\item It is straightforward that detection algorithms are expected to detect more localizable nodes when there are more anchors. However, only the localizable nodes detected by IMF increase monotonically with the increase of anchor nodes under fixed network density. That is, the detection capacities of TP and TE are additionally affected by the distribution of anchors. 
\end{itemize}


\subsection{Applicability of IMF in Harsh Scenarios}
To further show the stable performance of IMF, we visualize the detection result of different algorithms in harsh scenarios, i.e., networks with irregular anchor distribution or networks with holes. 

\begin{figure}[ht!]
	\centering
	\subfigure[]{
		\begin{minipage}[t]{0.46\linewidth}
			\centering
			\includegraphics[width=.99\textwidth]{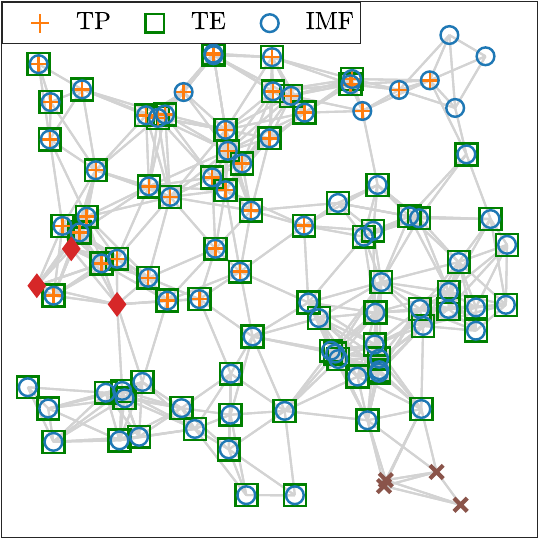}
			\label{fig:harsh_0}
		\end{minipage}
	}
	\subfigure[]{
		\begin{minipage}[t]{0.46\linewidth}
			\centering
			\includegraphics[width=.99\textwidth]{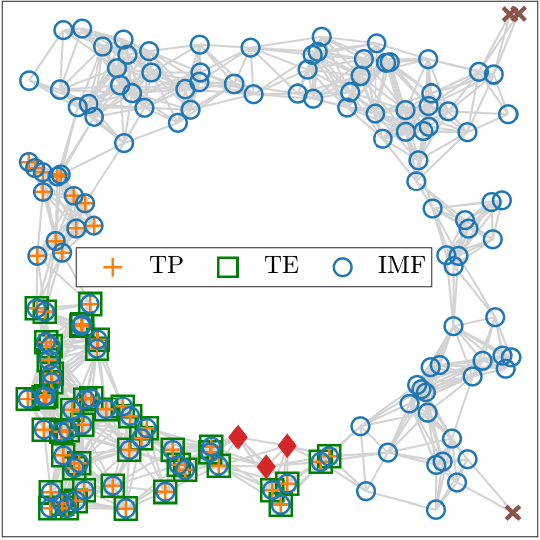}
			\label{fig:hole_0}
		\end{minipage}
	}\\
	\subfigure[]{
		\begin{minipage}[t]{0.46\linewidth}
			\centering
			\includegraphics[width=.99\textwidth]{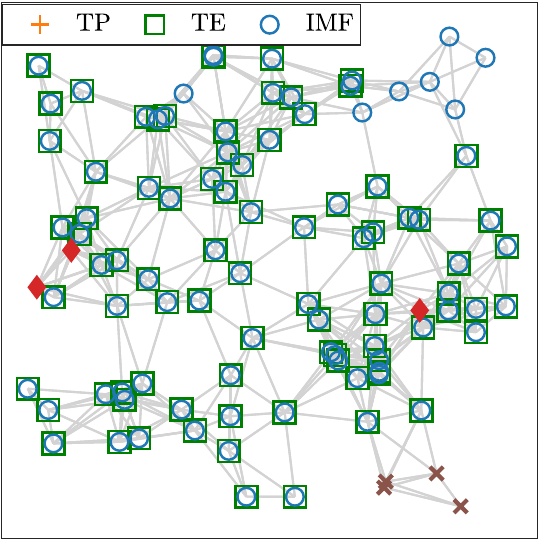}
			\label{fig:harsh_1}
		\end{minipage}
	}
	\subfigure[]{
		\begin{minipage}[t]{0.46\linewidth}
			\centering
			\includegraphics[width=.99\textwidth]{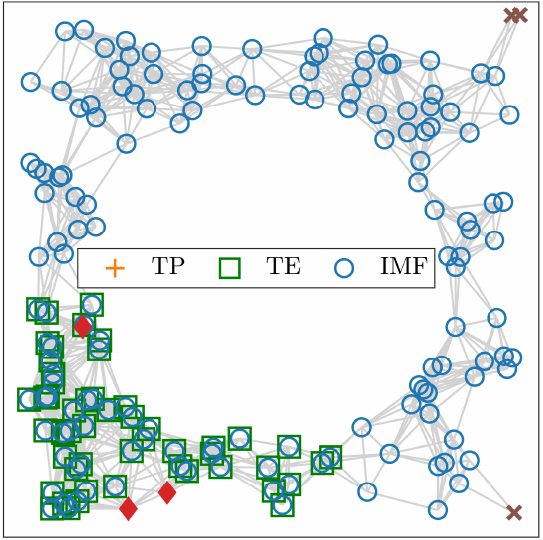}
			\label{fig:hole_1}
		\end{minipage}
	}\\
	\subfigure[]{
		\begin{minipage}[t]{0.46\linewidth}
			\centering
			\includegraphics[width=.99\textwidth]{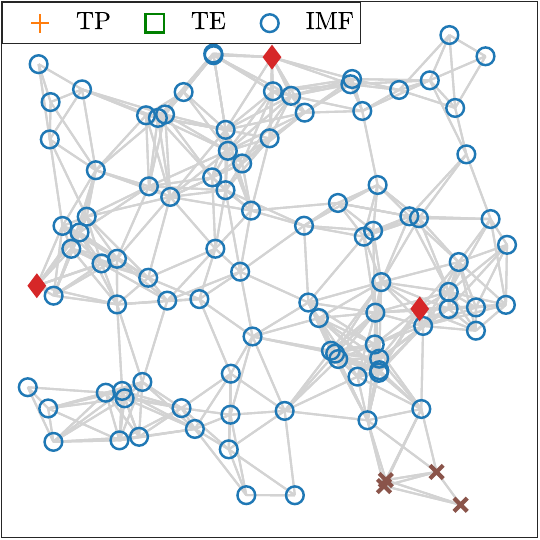}
			\label{fig:harsh_2}
		\end{minipage}
	}
	\subfigure[]{
		\begin{minipage}[t]{0.46\linewidth}
			\centering
			\includegraphics[width=.99\textwidth]{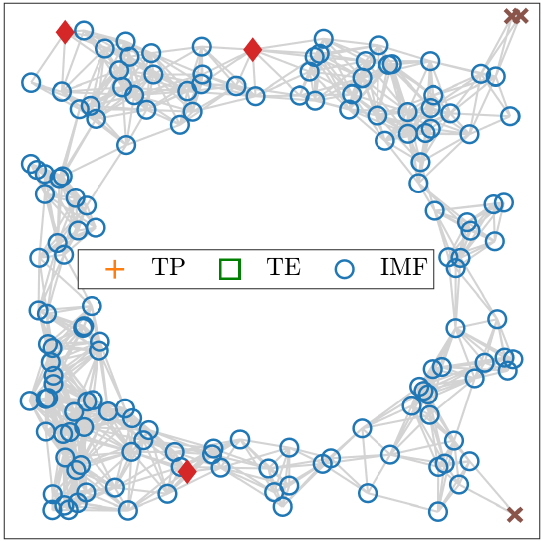}
			\label{fig:hole_2}
		\end{minipage}
	}
	\caption{The networks with different anchor distributions. Red diamonds represent anchors and brown cross markers represent the nodes that are not detected to be localizable by any algorithm. (a)-(b): 3 anchors have common neighbors; (c)-(d): 2 anchors have common neighbors; (e)-(f): no anchor nodes have common neighbors.)}
	\label{fig:show_harsh}
\end{figure}


Fig.~\ref{fig:show_harsh} shows the impact of anchor distribution. It can be seen that TP requires at least three anchors having common neighbors as in Fig.~\ref{fig:harsh_0} to startup and TE needs at least two anchors having common neighbors as in Fig.~\ref{fig:harsh_1} to startup. If the startup condition is not satisfied as in Fig.~\ref{fig:harsh_2}, TP and TE fail and detect no localizable nodes. However, IMF can detect most of the localizable nodes, even when anchors do not have any common neighbor.

In some applications, networks may be deployed with holes due to geographical obstacles or deployment requirements. To clarify the effect of holes, we first adjust the anchor distribution to ensure TP and TE can startup as in Fig.~\ref{fig:hole_0} and Fig.~\ref{fig:hole_1}. The results show that even when the startup conditions are met, TP and TE detect fewer localizable nodes than IMF. When anchors have no common neighbors, only IMF works as in Fig.~\ref{fig:hole_2}. The detailed detection results are given in Table~\ref{tab:number_localizable}.

\begin{table}[ht!]
	\caption{The number of detected localizable nodes.}
	\centering
	\begin{tabular}{p{0.5cm}|p{0.75cm}p{0.95cm}|p{0.75cm}p{0.95cm}|p{0.75cm}p{0.85cm}}
		\toprule
		~   & Fig.~\ref{fig:harsh_0} & Fig.~\ref{fig:hole_0} & Fig.~\ref{fig:harsh_1} & Fig.~\ref{fig:hole_1} & Fig.~\ref{fig:harsh_2} & Fig.~\ref{fig:hole_2} \\ \midrule
		Total & 97 & 154 & 97 & 154 & 97 & 154\\	
		TP  & 41 &  64 & 0  & 0 & 0 & 0\\
		TE & 86 &  51 &  86 & 51 & 0 & 0 \\
		\textbf{IMF} & \textbf{93} & \textbf{151}  & \textbf{93} & \textbf{151} & \textbf{93} & \textbf{151}\\
		\bottomrule
	\end{tabular}
	\label{tab:number_localizable}
\end{table}
\begin{figure*}[ht!]
	\centering
	\subfigure[Average Degree=8.]{
		\begin{minipage}[t]{0.19\linewidth}
			\centering
			\includegraphics[width=.99\textwidth]{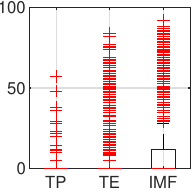}
			\label{fig:box_1}
		\end{minipage}
	}\hspace{-0.3cm}
	\subfigure[Average Degree=10.]{
		\begin{minipage}[t]{0.19\linewidth}
			\centering
			\includegraphics[width=.99\textwidth]{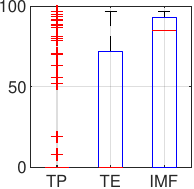}
			\label{fig:box_2}
		\end{minipage}
	}\hspace{-0.3cm}
	\subfigure[Average Degree=12.]{
		\begin{minipage}[t]{0.19\linewidth}
			\centering
			\includegraphics[width=0.99\textwidth]{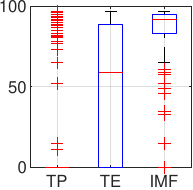}
			\label{fig:box_3}
		\end{minipage}
	}\hspace{-0.3cm}
	\subfigure[Average Degree=14.]{
		\begin{minipage}[t]{0.19\linewidth}
			\centering
			\includegraphics[width=0.99\textwidth]{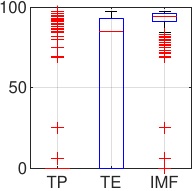}
			\label{fig:box_4}
		\end{minipage}
	}
	\hspace{-0.3cm}
	\subfigure[Average Degree=16.]{
		\begin{minipage}[t]{0.19\linewidth}
			\centering
			\includegraphics[width=0.99\textwidth]{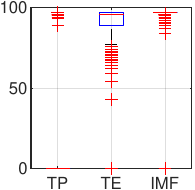}
			\label{fig:box_5}
		\end{minipage}
	}
	\caption{The distribution of detected localizable nodes under different average degrees.}
	\label{fig:box_localizability}
\end{figure*}

Moreover, Fig.~\ref{fig:box_localizability} shows the statistical results of repeated experiments. The five figures are statistical results for networks with average degrees \{8, 10, 12, 14, 16\} respectively.  In each figure, 1000 networks are generated randomly following the average node degree and run TP, TE, and IMF respectively. Each network contains 100 nodes and the distributions of the number of detected localizable nodes for TP, TE and IMF are compared using box graph. It can be observed that:
\begin{itemize}
	\item IMF always detects more localizable nodes than TP and TE, especially in sparse networks. When the average degree is 16, the median value of IMF reaches 97\%. 
	\item The median value of TP is 0 under each degree, which means that its startup condition is hard to be satisfied in random networks.
\end{itemize}

\section{Conclusion}
\label{sect:conclusion}
Determining localizability in BLL methods is crucially important due to its wide applications. In this paper, we summarize the gaps in existing theories and algorithms of BLL-localizability. To fill the knowledge gap, a sufficient condition and a necessary condition for BLL node localizability are proposed. Based on these conditions, algorithms to test BLL network localizability and to detect BLL node localizability are designed. We theoretically analyzed our main algorithm and discussed its broad prospects. Evaluations demonstrated the validity of the proposed IMF node localizability detection algorithm. 

In future work, promising directions include deriving the sufficient and necessary condition of node localizability for both NLL and BLL, and designing reliable algorithms guaranteeing the detection of all theoretically localizable nodes, especially in higher dimensions.

\bibliographystyle{unsrt}
\bibliography{IMF}

\end{document}